\documentclass[10pt, doublecolumn]{IEEEtran}
\usepackage{epsfig,latexsym}
\usepackage{float}
\usepackage{indentfirst}
\usepackage{amsmath}
\usepackage{amssymb}
\usepackage{times}
\usepackage{subfigure}
\usepackage{psfrag}
\usepackage{hyperref}
\usepackage{cite}
\usepackage{lastpage}
\usepackage{fancyhdr}
\usepackage{color}
 \usepackage{amsthm}
\usepackage{bigints}
\newtheorem{Lemma}{Lemma}
\newtheorem{lemma}[Lemma]{$\mathbf{Lemma}$}

\begin{document}
\title{ \huge{  Joint Power and Time Allocation for NOMA-MEC Offloading  }}

\author{ Zhiguo Ding,  Jie Xu,     Octavia A. Dobre,  and H. Vincent Poor  \thanks{   

   Z. Ding and H. V. Poor are  with the Department of
Electrical Engineering, Princeton University, Princeton,
USA.  Z. Ding  is  also  with the School of
Electrical and Electronic Engineering, the University of Manchester, Manchester,  UK. J. Xu is with   School of Information Engineering,
Guangdong University of Technology,  China.  
O. A. Dobre is with the Department of Electrical and Computer Engineering, Memorial University, 
St. John's, NL, Canada.
  
  }\vspace{-2em}} \maketitle

\begin{abstract} This paper considers   non-orthogonal multiple access (NOMA) assisted  mobile edge computing (MEC), where  the power and time allocation is jointly optimized to reduce  the energy consumption of  offloading. Closed-form expressions for the optimal power and time allocation solutions are obtained and used to establish the conditions for determining whether conventional   orthogonal multiple access (OMA), pure NOMA or hybrid NOMA  should be used for MEC offloading.    
\vspace{-1em}
\end{abstract} 

\section{Introduction}
 Both non-orthogonal multiple access (NOMA) and mobile edge computing (MEC) have been recognized as    important   techniques in future wireless networks \cite{jsacnomaxmine,6553297}.  Sophisticated optimization frameworks developed in  \cite{8269088, 8267072}   show that     by applying NOMA to MEC, not only can severe   delay   be avoided, but also      energy consumption can be reduced, although the comparisons between NOMA and orthogonal multiple access (OMA) in \cite{8269088, 8267072} rely on simulation. Insightful analytical results   developed in   \cite{MECding} confirmed the advantages  of NOMA-MEC offloading, by using fixed bandwidth allocation.

This letter studies the impact of NOMA on energy-efficient  MEC offloading,  by focusing on the fundamental two-scheduled-user case in order to obtain an insightful understanding of  NOMA-MEC. The existing studies in  \cite{8269088, 8267072, MECding} consider  two  offloading strategies only,      OMA and pure NOMA (i.e.,  both the users offload all of their tasks at the same time). However, there is a third strategy, termed    hybrid NOMA   in this paper, i.e.,  a user  can first offload  parts of its task by using a time slot allocated to another user and then offload the remainder  of its task during a time slot solely occupied by itself.  The performance of the three   strategies is studied in this paper, where closed-form expressions for the optimal  time and power allocation solutions   are obtained, by applying geometric programming (GP).  These closed-form solutions not only  facilitate   low-complexity resource allocation, but also reveal  important  properties of NOMA-MEC offloading. For example, by using the obtained closed-form solutions,  hybrid-NOMA-MEC can be   proved to be superior to OMA-MEC when users have demanding latency requirements for their task offloading, whereas OMA-MEC is preferred if a user's task is delay tolerant. It is worth pointing out that the pure NOMA   strategy     is not preferred for either of the two situations. 

\section{System Model} 
Consider an   MEC offloading  scenario, in which $K$ users with different quality of service (QoS) requirements   communicate with one access point with an integrated MEC server.  Because of their limited computational capabilities,  it is assumed that the   users  choose to offload   their  computationally intensive, latency-critical, and inseparable tasks to the  server.

Each user's task is characterized by the parameter pair $\{N_k,   \beta_k\}$, $k=1,\ldots, K$, which is defined as follows:
\begin{itemize}
\item $N_k$ denotes the number of nats contained in a task;
\item $D_k$ denotes the computation deadline of a task. 
\end{itemize} 

Without loss of generality, assume that $N_k= N$, $1\leq k\leq K $, and the users are ordered according to their computation deadlines, i.e., $ D_1\leq  \cdots \leq D_K$. 
To reduce the system complexity, it is further assumed that the MEC server schedules only two users, user $m$ and user $n$, $m\leq n$,  to be served at the same resource block.  Note that scheduling two users to perform NOMA is also aligned with how NOMA is implemented in LTE-A \cite{3gpp1}. To better illustrate the benefit of NOMA, OMA-MEC is illustrated first. 
 
If OMA is used, each user is allocated a dedicated time slot for offloading\footnote{In this paper, the time   and the energy costs for the server to send the outcomes of the tasks to the users    are omitted, since the size of the outcomes is typically very small.   The energy consumption for the computation at the server is also omitted, as the server is not energy constrained.    }.  Since user $m$ has a more demanding deadline than user $n$, user $m$   is served first. Therefore the users'  transmit powers, denoted by $P^{\text{OMA}}_m$ and $P^{\text{OMA}}_n$, need to satisfy  $
  D_m\ln(1+P^{\text{OMA}}_m|h_m|^2)=N$ and $ (D_n-D_m)\ln(1+P^{\text{OMA}}_n|h_n|^2)=N$, respectively, 
where $h_i$ denotes user $i$'s channel gain, $i= m,n$.

By using the principle of NOMA,   the two users can offload their tasks simultaneously during $D_m$ to the server. It is important to point out that user $m$ experiences the same  performance as in OMA if its message is decoded at the second stage of successive interference cancelation (SIC) and    user $n$'s data rate during $D_m$ is constrained as $
R_n\leq \ln \left(1+\frac{P_{n,1}|h_n|^2}{P_m^{\text{OMA}}|h_m|^2+1}\right)$,
where $P_{n,1}$ denotes the power used by user $n$ during $D_m$. 
 
As pointed out in \cite{MECding}, user $n$ needs to consume  more energy in NOMA  than   in OMA if the user completely relies on $D_m$. Therefore, hybrid NOMA is considered, i.e., user $n$ shares $D_m$ with user $m$, and then continuously transmits for another time interval, denoted by $T_n$, after $D_m$.   Denote the power used by user $n$ during   $T_n$ by   $P_{n,2}$.  As user $m$ experiences the same as in OMA, we     focus only on user $n$'s performance in this letter.

\section{ NOMA-Assisted MEC Offloading}
The problem for minimizing the energy consumption of   NOMA-MEC offloading can be formulated as follows: 
\begin{subequations}\label{5}
\begin{eqnarray}\label{5a}&
\underset{T_n, P_{n,1},P_{n,2}}{\min}\quad& D_mP_{n,1}+T_nP_{n,2}\\ \label{5b}
&s.t. \quad& D_m\ln\left(1+\frac{P_{n,1}|h_n|^2}{P_m^{\text{OMA}}|h_m|^2+1}\right)\\\nonumber&&+T_n\ln\left(1+|h_n|^2P_{n,2}\right)\geq N\\ \label{5c} &&0\leq T_n\leq D_n-D_m\\ && P_{n,i}\geq 0, \forall i \in\{1,2\}.
\end{eqnarray}
\end{subequations}
The objective function \eqref{5a} denotes user $n$'s   energy consumption for   MEC offloading, \eqref{5b} denotes the rate constraint to ensure that user $n$'s $N$ nats are offloaded within $D_m+T_n$, and \eqref{5c} denotes the deadline constraint, i.e., $T_n+D_m\leq D_n$. It is worth noting that the benefit of using NOMA is obvious for the case of $D_n=D_m$, where the power required by the OMA case becomes infinite while the power in NOMA is finite.   

In the first two subsections of this section, we  will focus on the scenario where $D_n< 2D_m$, in order to avoid the trivial case with OMA solutions. 
In particular, we first obtain the optimal solutions for $P_{n,1}$ and $P_{n,2}$ as explicit functions of  $T_n$ by applying GP,  and then find the optimal solution of $T_n$. The scenario   $D_n\geq 2D_m$ is also  discussed at the end of this section. 
 
\subsection{Finding the Optimal Solutions for $P_{n,1}$ and $P_{n,2}$}
In order to make GP  applicable, the objective function and the constraints in \eqref{5} need to be transformed as follows. By using the fact $D_m\ln(1+P^{\text{OMA}}_m|h_m|^2)=N$, constraint \eqref{5b} can be simplified as follows: 
\begin{align}\label{rate constraint}
\ln\left(1+e^{-\frac{N}{D_m}}|h_n|^2P_{n,1}\right)^{D_m}\left(1+|h_n|^2P_{n,2}\right)^{T_n}\geq N. 
\end{align}

Define $x_1= 1+e^{-\frac{N}{D_m}}|h_n|^2P_{n,1}$ and $x_2=1+|h_n|^2P_{n,2}$. Problem  \eqref{5} is transformed   to the following equivalent form:   
\begin{subequations}\label{6}
\begin{eqnarray}\label{6a}&
\underset{T_n, x_1,x_2}{\min}\quad& D_me^{\frac{N}{D_m}} x_1  +T_n\left(x_2-1\right) \\ 
&s.t. \quad& e^{N}x_1^{-D_m}   x_2^{-T_n}\leq 1\\  &&0\leq T_n\leq D_n-D_m\\ && x_i\geq 1, \forall i \in\{1,2\}.
\end{eqnarray}
\end{subequations}

Define $y_i=\ln x_i$, $i=1,2$. By fixing $T_n$,   problem   \eqref{6} can be transformed  to the following equivalent form: 
\begin{subequations}\label{7}
\begin{eqnarray}\label{7a}&
\underset{ y_1,y_2}{\min}\quad& D_me^{\frac{N}{D_m}} e^{y_1}  +T_ne^{y_2} \\ 
&s.t. \quad&  e^{-D_my_1 -T_ny_2+N}\leq 1\\ && y_i\geq 0, \forall i \in\{1,2\}. 
\end{eqnarray}
\end{subequations}

By treating   problem \eqref{7} as a special case of GP and applying  logarithm to \eqref{7},   the Karush-Kuhn-Tucker (KKT) conditions can be applied to find the optimal solution as follows:
  \begin{eqnarray}
\left\{\begin{array}{rl}  \frac{D_me^{\frac{N}{D_m}} e^{y_1} }{D_me^{\frac{N}{D_m}} e^{y_1}  +T_ne^{y_2}} -\lambda_1-\lambda_3 D_m&=0 \\  
    \frac{T_ne^{y_2} }{D_me^{\frac{N}{D_m}} e^{y_1}  +T_ne^{y_2}} -\lambda_2-\lambda_3 T_n&=0
    \\ 
    N- D_my_1-T_ny_2&\leq 0 \\ 
     \lambda_3\left( -D_my_1-T_ny_2+N  \right)&=0 \\
     -y_i&\leq 0, \forall i \in\{1,2\}\\
     \lambda_iy_i&=0, \forall i \in\{1,2\}\\
     \lambda_i&\geq 0, \forall i \in\{1, 2,3\}
     \end{array}\right.\hspace{-1em},
\end{eqnarray}
where $\lambda_i$ are   Lagrange multipliers. The optimal solutions of $P_{n,1}$ and $P_{n,2}$ can be obtained as in the following lemma.  
\begin{lemma}\label{lemma1}
Assume $D_n<2D_m$. The optimal solutions for $P_{n,1}$ and $P_{n,2}$ in   problem \eqref{5} can be expressed as the following closed-form functions of $T_n$: 
  \begin{eqnarray}\label{eqlema}
\left\{\begin{array}{rl}  
P_{n,1}^* = &|h_n|^{-2}e^{\frac{N}{D_m}}\left(e^{\frac{N(D_m-T_n)}{D_m(D_m+T_n)}}-1\right)\\ 
P_{n,2}^* = & |h_n|^{-2}\left(e^{\frac{N(D_m-T_n)}{D_m(D_m+T_n)} +\frac{N}{D_m}}-1\right)
     \end{array}\right..
\end{eqnarray}
\end{lemma}
 \begin{proof}
 Please refer to the appendix. 
 \end{proof} 
 
 \subsection{Finding the Optimal Solution for $T_n$}\label{subsection t}
By substituting the optimal solution obtained in Lemma~\ref{lemma1} into problem \eqref{5},  the original problem  can be written in an equivalent form as follows:
\begin{align}\label{gtn}
\underset{T_n}{\min}\quad&  g_{T_n}\triangleq D_m \left(e^{y_1^*}- 1\right)e^{\frac{N}{D_m}}+T_n\left(e^{y_2^*}-1\right),\\\nonumber
s.t. \quad& T_n\leq D_n-D_m,
\end{align}
where $g_{T_n}$ is the  energy consumption normalized by omitting the constant $|h_n|^{-2}$   in the objective function \eqref{5a}.   Note that both $y_1^*$ and $y_2^*$ are functions of $T_n$ as defined in \eqref{13}.  

The derivative of $g_{T_n}$ with respect to $T_n$ can be expressed  as follows:
\begin{align}
\frac{ dg_{T_n}}{dT_n}=& D_me^{\frac{N}{D_m}}  e^{y_1^*} \frac{(-2N)}{(D_m+T_n)^2}+\left(e^{y_2^*}-1\right) \\\nonumber &+T_ne^{y_2^*}\frac{(-2N)}{(D_m+T_n)^2}.
\end{align}
Recall that $y_2^* = y_1^*+\frac{N}{D_m}$. Therefore, the derivative of $g_{T_n}$ can be rewritten as follows:
\begin{align}
\frac{ dg_{T_n}}{dT_n}=& D_m  e^{y_2^*} \frac{(-2N)}{(D_m+T_n)^2}+\left(e^{y_2^*}-1\right) \\\nonumber &+T_ne^{y_2^*}\frac{(-2N)}{(D_m+T_n)^2} \\\nonumber =&e^{y_2^*}\left(1-\frac{2N}{D_m+T_n}\right)-1.
\end{align}
Further, recall  that $y_2^*= \frac{N(D_m-T_n)}{D_m(D_m+T_n)} +\frac{N}{D_m}=\frac{2N}{D_m+T_n}$. Thus, the derivative of $g_{T_n}$ can be expressed as follows:
\begin{align}
\frac{ dg_{T_n}}{dT_n} =&g_x\left(\frac{2N}{D_m+T_n}\right),
\end{align}
where 
\begin{align}
g_x(x)\triangleq e^{x}\left(1-x\right)-1.
\end{align}
 $g_x(x)$ is a monotonically non-increasing  function since $\frac{dg_x(x)}{dx}=-xe^{-x}\leq 0$ for $x\geq 0$. Therefore, $\frac{ dg_{T_n}}{dT_n}\leq 0$ since 
\begin{align}
\frac{ dg_{T_n}}{dT_n} \leq g_x\left(0\right)=0,
\end{align}
which means that $g_{T_n}$ is   monotonically non-increasing. Hence,  
 the optimal solution of $T_n$ for problem \eqref{5} is given by
\begin{align}
T_n^* = D_n-D_m .
\end{align} 
 It is worth pointing out that $T_n^* <D_m$, since the case $D_n<2D_m$ is considered  in this subsection. 
  
\subsection{Remarks and Discussions }\label{remarks}
\subsubsection{For the superiority of NOMA over OMA} we can   show that OMA cannot outperform NOMA, as presented   in the following. The energy consumption gap       between NOMA-MEC and OMA-MEC is given by
\begin{align}
\Delta \triangleq & D_m \left(e^{y_1^*}- 1\right)e^{\frac{N}{D_m}}|h_n|^{-2}+T_n\left(e^{y_2^*}-1\right)|h_n|^{-2} \\\nonumber &- T_n  \left(e^{\frac{N}{T_n}}-1\right)|h_n|^{-2}.
\end{align}
By using \eqref{13}, the gap can be further expressed as follows: 
\begin{align}
|h_n|^{2}\Delta \triangleq & D_m e^{y_2^*}(D_m+T_n) - D_me^{\frac{N}{D_m}} - T_ne^{\frac{N}{T_n}}\\\nonumber =&e^{\frac{2N}{D_m+T_n}}(D_m+T_n) - D_me^{\frac{N}{D_m}} - T_ne^{\frac{N}{T_n}} =f_{T_n}(T_n). 
\end{align}
As shown in \eqref{l1},  $f_{T_n}(T_n)\leq 0$, which means that the use of NOMA outperforms or at least yields the same performance as OMA, under the condition   $D_n<2D_m$. 

\subsubsection{For the case $D_n\geq 2 D_m$}\label{subsectionx}  this case corresponds to a scenario in which user $n$ has less demanding latency requirements. Compared to the case $D_n<2D_m$, $T_n$ can be larger than $D_m$ for the case  $D_n\geq 2 D_m$,  since $T_n= D_n-D_m$.  In this case, OMA yields the best performance, as shown in the following. 
Since  the hybrid NOMA solutions in Lemma \ref{lemma1}  are feasible only if $T_n<D_m$ and the energy consumption of hybrid NOMA, i.e.,  $g_{T_n}$ in \eqref{gtn}, is a monotonically non-increasing function of $T_n$, $g_{T_n}$ is always strictly lower bounded by
\begin{align}\label{lower bound}
D_m |h_n|^{-2}\left(e^{\frac{N }{ D_m}  }-1\right).
\end{align}
On the other hand,   the lower bound in \eqref{lower bound} can be achieved  by OMA when   $D_n\geq 2D_m$, i.e., the solution obtained with $\lambda_1\neq0$, $\lambda_2=0$ and $T_n=D_m$, as shown in \eqref{l2=0}. In other words, when  $D_n\geq 2 D_m$, OMA requires  less energy consumption than hybrid NOMA. Furthermore, OMA can also  outperform pure NOMA since
\begin{align}\nonumber 
\frac{E_{\text{OMA}} -
E_{\text{NOMA}}}{ |h_n|^{-2}}   \underset{(a)}{\leq}&
 D_m \left(e^{\frac{N}{D_m}}-1\right)   - D_m e^{\frac{N}{D_m}}\left(e^{\frac{N}{D_m}}-1\right)   \\  \label{stepx2}&=-D_m\left(e^{\frac{N}{D_m}}-1\right)^2\leq 0,
\end{align}
where step (a) is due to   the fact that the minimal energy required by OMA is no less than that in \eqref{lower bound}. Therefore, it is concluded that   OMA outperforms hybrid NOMA and pure NOMA when $D_n\geq 2 D_m$. This conclusion  is reasonable, since a more relaxed deadline makes it possible to   use  only the  interference-free time slot ($D_n-D_m$) for offloading.

\section{Numerical Results}
In this section, the performance of the proposed NOMA-MEC scheme is evaluated  via  simulation results, where the normalized energy consumption in \eqref{gtn} is used.  As can be observed from  Fig. \ref{fig1}, the use of NOMA-MEC can yield a significant performance gain over OMA-MEC, particularly when $D_n$ is small. This is because   OMA-MEC  relies on the short period $(D_n-D_m)$ for offloading. Take $D_n\rightarrow D_m$ as an example. $(D_n-D_m)$ becomes close to zero, and hence the energy consumed by OMA-MEC becomes prohibitively large, as shown in the figure. On the other hand, NOMA-MEC uses not only $(D_n-D_m)$ but also $D_m$ for offloading, which makes the energy consumed by NOMA-MEC   more stable. 
  \begin{figure}[!htbp]\centering\vspace{-1em}
    \epsfig{file=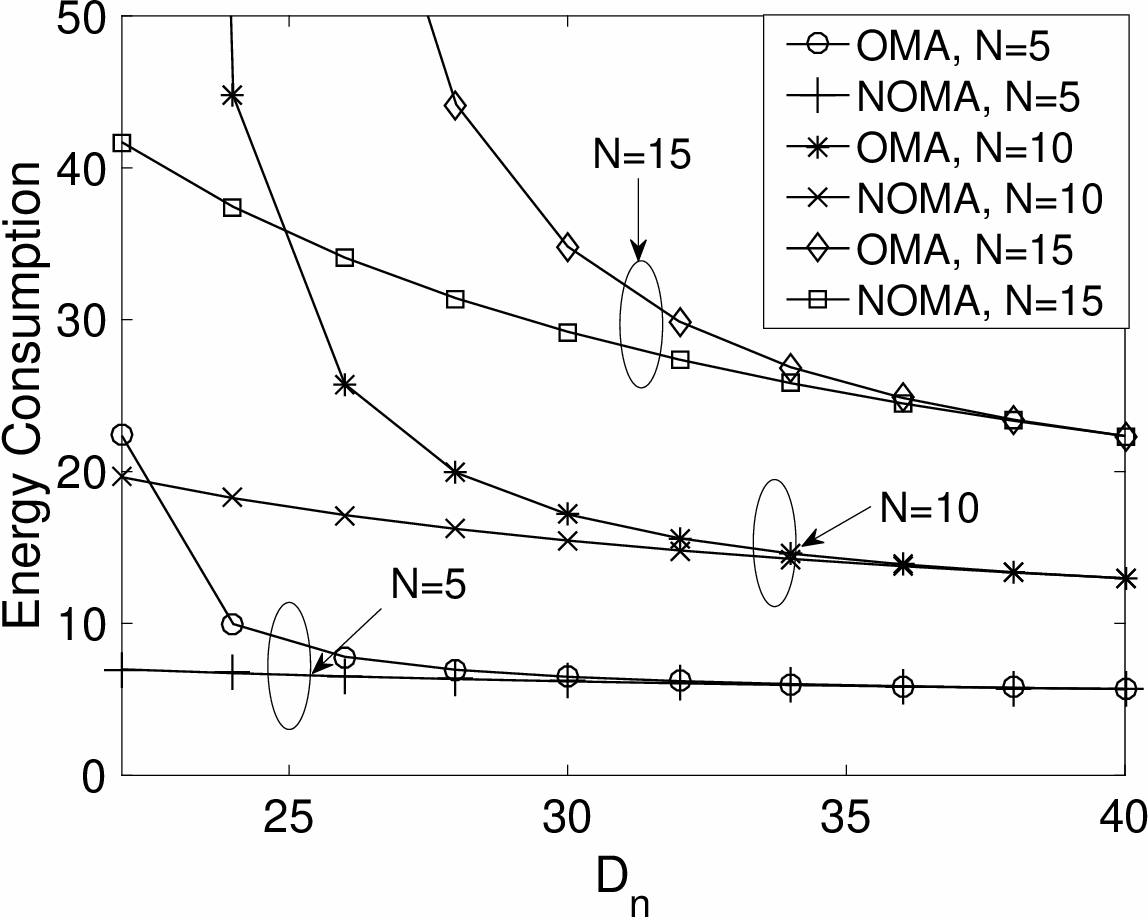, width=0.35\textwidth, clip=}\vspace{-1em}
\caption{ Performance comparison between NOMA-MEC and OMA-MEC, where  $D_m=20$. \vspace{-0.2em} }\label{fig1}
\end{figure}
  \begin{figure}[!htbp]\centering\vspace{-1em}
    \epsfig{file=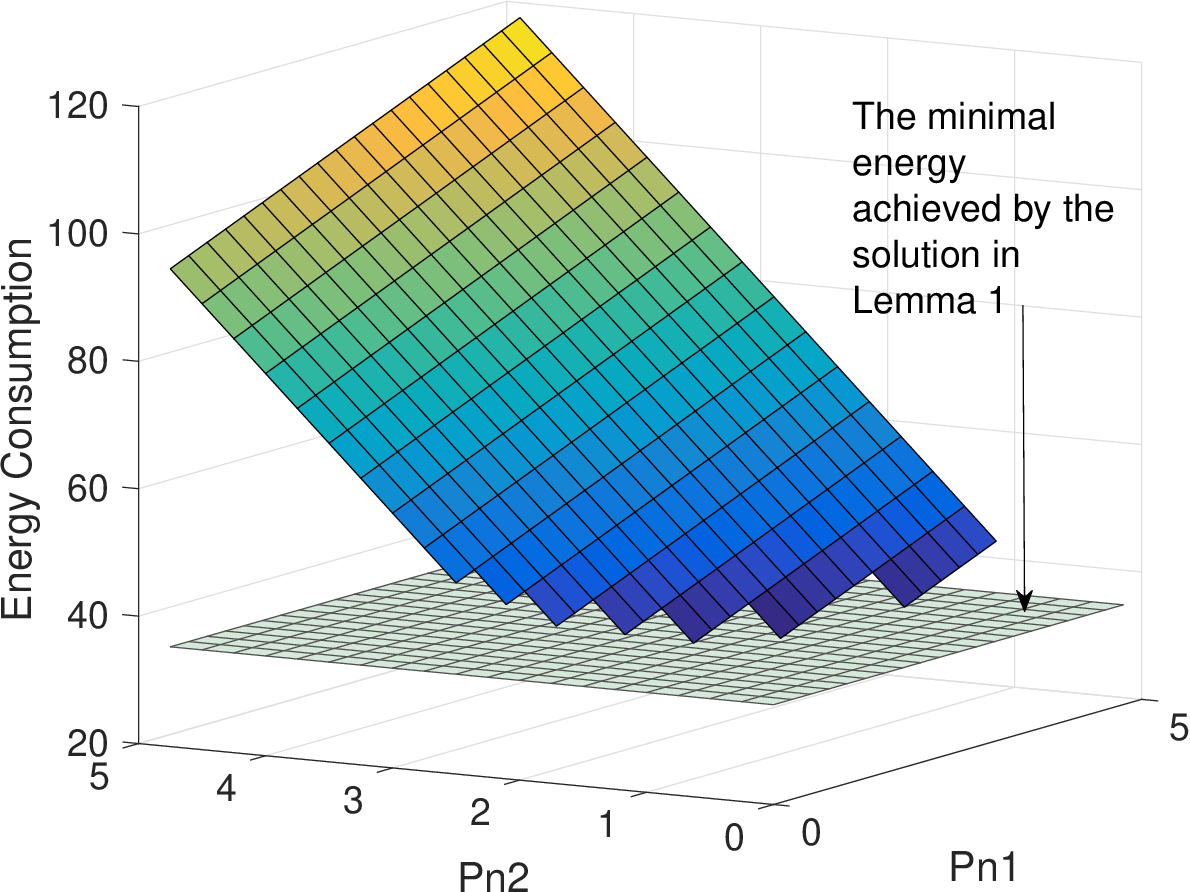, width=0.4\textwidth, clip=}\vspace{-1em}
\caption{Optimality of the solutions obtained in Lemma \ref{lemma1}, where $N=15$, $D_m=20$, and $T_n=\frac{1}{4}D_m$. \vspace{-0.2em} }\label{fig2}
\end{figure}
  \begin{figure}[!htbp]\centering\vspace{-1em}
    \epsfig{file=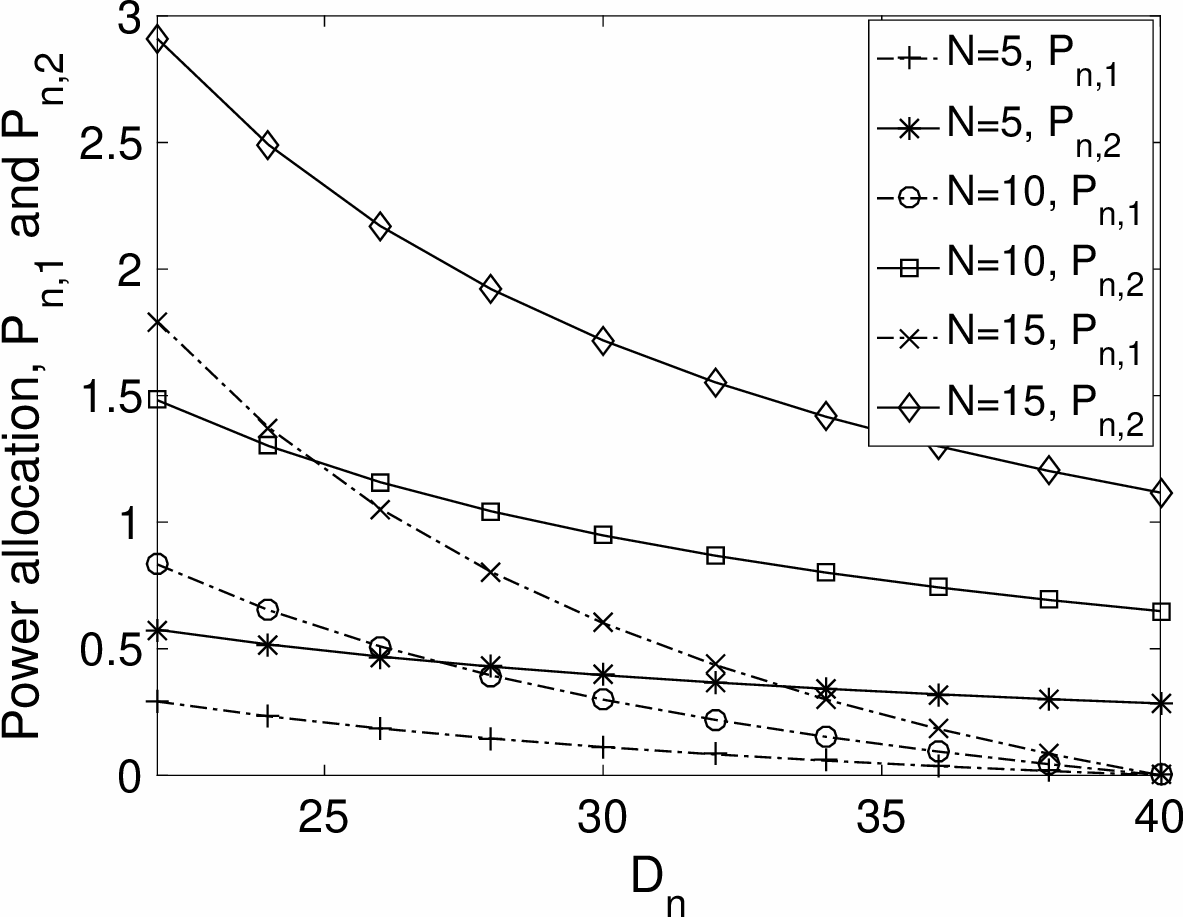, width=0.35\textwidth, clip=}\vspace{-1em}
\caption{The performance of NOMA-MEC, where   $D_m=20$. \vspace{-1em} }\label{fig3}
\end{figure}

To better illustrate the optimality of the solutions obtained in Lemma \ref{lemma1}, the energy consumption is shown  as a function of different choices of $(P_{n,1}, P_{n,2})$ in Fig. \ref{fig2}.   The figure clearly demonstrates that among all the possible power allocation choices, the one provided in Lemma \ref{lemma1} yields the lowest energy consumption. As discussed in Section \ref{remarks}, the performance of NOMA and OMA becomes quite similar when   $D_n$ becomes large, which is  confirmed by Fig. \ref{fig1}, while further  details about this aspect are provided  in Fig. \ref{fig3}. As can be seen from this figure, when $D_n$ increases,  the power allocated to $D_m$ approaches   zero, which means that hybrid NOMA is degraded relative to OMA, as pointed out in Section \ref{remarks}. 
\vspace{-1em}
\section{Conclusions}
 In this paper, the principle of  NOMA has been  applied to MEC, and   optimal solutions for the power and time allocation have been obtained  by applying GP.    Analytical and simulation results   have also been  provided to demonstrate the superior performance of NOMA-MEC over OMA-MEC.  
\vspace{-1em}
\appendices
\section{Proof of Lemma \ref{lemma1}}
The proof of the lemma can be completed by studying the possible choices of $\lambda_i$, $i=1, 2,3$, and showing that the solutions for the case with $\lambda_i=0$, $\forall i\in\{1,2\}$, yield the smallest energy consumption.  
\subsubsection{Hybrid NOMA ($\lambda_i=0$, $\forall i\in\{1,2\}$)} since $\lambda_i=0$, $\forall i\in\{1,2\}$,  $y_i>0$ and hence $P_{n,1}$ and $P_{n,2}$ are non-zero, which is the reason why this case is termed hybrid NOMA. For this case, we can show that $\lambda_3\neq 0$ as follows. If $\lambda_3=0$, the KKT conditions lead to the following two equations: 
  \begin{eqnarray}
\left\{\begin{array}{rl}  
  \frac{D_me^{\frac{N}{D_m}} e^{y_1} }{D_me^{\frac{N}{D_m}} e^{y_1}  +T_ne^{y_2}} &=0 \\  
    \frac{T_ne^{y_2} }{D_me^{\frac{N}{D_m}} e^{y_1}  +T_ne^{y_2}}  &=0
     \end{array}\right.,
\end{eqnarray}
which cannot be true. Therefore,   $\lambda_3\neq 0$ follows, which means that the KKT conditions can be rewritten as follows: 
  \begin{eqnarray}
\left\{\begin{array}{rl}  
  \frac{ e^{\frac{N}{D_m}} e^{y_1} }{D_me^{\frac{N}{D_m}} e^{y_1}  +T_ne^{y_2}} -\lambda_3  &=0 \\  
    \frac{ e^{y_2} }{D_me^{\frac{N}{D_m}} e^{y_1}  +T_ne^{y_2}} -\lambda_3  &=0
    \\  
      \left( -D_my_1-T_ny_2+N  \right)&=0\\ 
      y_i&> 0, \forall i \in\{1,2\}
     \end{array}\right..
\end{eqnarray}
With some algebraic manipulations, the optimal solutions for $y_1$ and $y_2$ can be obtained as follows: 
  \begin{eqnarray}\label{13}
\left\{\begin{array}{rl}  
y_1^* =& \frac{N(D_m-T_n)}{D_m(D_m+T_n)}\\ 
y_2^* =& \frac{N(D_m-T_n)}{D_m(D_m+T_n)} +\frac{N}{D_m}
     \end{array}\right..
\end{eqnarray}
Since $D_n<2D_m$, $T_n\leq D_n-D_m<D_m$, and the solutions   $y_i^*$'s satisfy the constraints $y_i>0$, which mean that the solutions shown in \eqref{13} are feasible. 
With the power allocation solutions in \eqref{13}, the overall energy consumption is given by
\begin{align}
E_{\text{H-NOMA}} =& D_m |h_n|^{-2}e^{\frac{N}{D_m}}\left(e^{\frac{N(D_m-T_n)}{D_m(D_m+T_n)}}-1\right) \\\nonumber &+T_n |h_n|^{-2}\left(e^{\frac{N(D_m-T_n)}{D_m(D_m+T_n)} +\frac{N}{D_m}}-1\right). 
\end{align}
%

\subsubsection{Pure NOMA ($\lambda_1=0$ and $\lambda_2\neq 0$)} since $\lambda_1=0$ and $\lambda_2\neq 0$, we have $y_1\neq0$ and $y_2=0$, and hence $P_{n,1}\neq0$ and $P_{n,2}=0$, which is the reason to term this case   pure NOMA. Since $y_2=0$ corresponds to an extreme situation in which all the power is allocated to $D_m$,  the use of the rate constraint in \eqref{rate constraint} yields the following choice of $P_{n,1}$:
\begin{align}
\tilde{P}_{n,1}^{*} = \left(e^{\frac{N}{D_m}}-1\right) e^{\frac{N}{D_m}} |h_n|^{-2},
\end{align} 
which means that the overall energy consumption becomes 
\begin{align}\label{l2=0}
E_{\text{NOMA}}= D_m\left(e^{\frac{N}{D_m}}-1\right) e^{\frac{N}{D_m}} |h_n|^{-2}. 
\end{align}

\subsubsection{OMA ($\lambda_1\neq0$ and $\lambda_2= 0$)} since $\lambda_1\neq0$ and $\lambda_2= 0$, we have $y_1=0$ and $y_2\neq 0$, and hence $P_{n,1}=0$ and $P_{n,2}\neq0$, which is the reason to term this case as OMA. Since  all the power is allocated to $T_n$, the use of  the rate constraint in \eqref{rate constraint} yields the following choice of $P_{n,2}$:
\begin{align}
\tilde{P}_{n,2}^* = \left(e^{\frac{N}{T_n}}-1\right)   |h_n|^{-2},
\end{align} 
which means that the overall energy consumption becomes 
\begin{align}
E_{\text{OMA}}=T_n \left(e^{\frac{N}{T_n}}-1\right)   |h_n|^{-2}. 
\end{align}

\subsubsection{Comparisons among the three cases} in the following, we can show that hybrid NOMA requires  the smallest energy. As discussed in Subsection \ref{subsection t}, the overall energy is a monotonically non-increasing function of $T_n$ when $\lambda_i=0$, $\forall i \in\{1,2\}$. Therefore, $E_{\text{H-NOMA}}$   is upper bounded by 
\begin{align}\label{l2}
E_{\text{H-NOMA}} \leq& D_m |h_n|^{-2}e^{\frac{N}{D_m}}\left(e^{\frac{N}{D_m}}-1\right) = E_{\text{NOMA}},
\end{align}
since $T_n\geq0$. Hence,  the use of hybrid NOMA requires  less energy consumption than pure NOMA. 

The difference between $E_{\text{H-NOMA}}$ and $E_{\text{OMA}}$ can be expressed as follows:
\begin{align}
\frac{E_{\text{H-NOMA}}- E_{\text{OMA}}}{|h_n|^{-2}}=D_m e^{\frac{N}{D_m}}\left(e^{\frac{N(D_m-T_n)}{D_m(D_m+T_n)}}-1\right) \\\nonumber +T_n  \left(e^{\frac{2N}{(D_m+T_n)} }-1\right)-T_n \left(e^{\frac{N}{T_n}}-1\right)   
= f_{T_n}(T_n),
\end{align}
where $f_{T_n}(x)$ is defined as follows: 
\begin{align}
f_{T_n}(x)\triangleq (D_m +x)e^{\frac{2N }{(D_m+x)}}-D_m e^{\frac{N}{D_m}}   -x  e^{\frac{N}{x}} . 
\end{align}
Note that $f_{T_n}(x)$ is a monotonically non-decreasing  function for $x<D_m$, as shown in the following. The derivative of $f_{T_n}(x)$ is given by
\begin{align}
\frac{df_{T_n}(x)}{dx}   =&e^{\frac{2N}{D_m+x}} \left(1-\frac{2N}{D_m+x}\right) - e^{\frac{N}{x}}\left(1 - \frac{N}{x}\right). 
\end{align}
Now define $f_y(y) = e^{\frac{N}{y}}\left(1 - \frac{N}{y}\right)$, and the derivative $f_{T_n}(x)$ can be expressed as follows:
\begin{align}
\frac{df_{T_n}(x)}{dx}   =& f_y\left(\frac{D_m+x}{2}\right)-f_y\left(x\right). 
\end{align}
Note that $f_y(y)$ is a monotonically  increasing function since $\frac{df_y(y)}{dy}=\frac{N^2e^{\frac{N}{y}}}{y^3}> 0$. Since $x<D_m$,   $\frac{D_m+x}{2}> x$.   Therefore,  the derivative $f_{T_n}(x)$ is non-negative, i.e.,
\begin{align}
\frac{df_{T_n}(x)}{dx}   =&f_y\left(\frac{D_m+T_n}{2}\right) -f_y(T_n)\geq 0,
\end{align} 
which means that $f_{T_n}(x)$ is a monotonically non-decreasing  function. Since $T_n<D_m$, we   have  
\begin{align} \label{l1}
\frac{E_{\text{H-NOMA}}- E_{\text{OMA}}}{|h_n|^{-2}} =f_{T_n}(T_n) \leq  & f_{T_n}(D_m)=0. 
\end{align}
Combining \eqref{l2} and \eqref{l1}, hybrid NOMA, i.e., the solutions obtained with $\lambda_i=0$, $\forall i\in\{1,2\}$, yields the smallest energy consumption. By using $y_i^*$ in \eqref{13}, the required powers during $D_m$ and $T_n$ can be obtained, and the proof is complete.
 \vspace{-1em}
   \bibliographystyle{IEEEtran}
\bibliography{IEEEfull,trasfer}

  \end{document}